\documentclass[12pt]{article}
\usepackage{amsmath}
\usepackage{amssymb}
\usepackage{amsfonts}
\usepackage{latexsym}
\usepackage{color}

\catcode `\@=11 \@addtoreset{equation}{section}

\catcode `\@=12

\newcommand{\be}{\begin{equation}}
\newcommand{\en}{\end{equation}}
\newcommand{\bea}{\begin{eqnarray}}
\newcommand{\ena}{\end{eqnarray}}
\newcommand{\beano}{\begin{eqnarray*}}
\newcommand{\enano}{\end{eqnarray*}}
\newcommand{\bee}{\begin{enumerate}}
\newcommand{\ene}{\end{enumerate}}

\newcommand{\mc}{\mathcal}

\newcommand{\G}{{\mc G}}

\newcommand{\Hil}{{\cal H}}

\newcommand{\F}{{\cal F}}
\newcommand{\Lc}{{\cal L}}
\newcommand{\D}{{\cal D}}
\newcommand{\C}{{\cal C}}

\newcommand{\Sc}{{\cal S}}

\newcommand{\A}{\mathfrak{A}}

\newcommand{\1}{1 \!\!\! 1}
\newtheorem{thm}{Theorem}

\newtheorem{prop}[thm]{Proposition}
\newtheorem{defn}[thm]{Definition}

\newenvironment{proof}{\noindent {\bf Proof:}}{\hfill$\Box$}
\newcommand{\ip}[2]{\langle {#1},{#2}\rangle}
\newcommand{\vp}{\varphi}

\newcommand{\bedefi}{\begin{defn}$\!\!${\bf }$\;$\rm }
\newcommand{\findefi}{\end{defn}}
\newcommand{\mb}{\mathbb}
\newcommand{\myalert}[1]{{\color{magenta}#1}}

\def\spn{{\sf span\,}}
\newcommand{\e}{{\sf e}}

\def\*{^*}
\begin{document}

\thispagestyle{empty}

\vspace*{1cm}

\begin{center}
{\Large \bf Gibbs states defined by biorthogonal sequences}   \vspace{2cm}\\

{\large F. Bagarello}
%\footnote[1]{ Dipartimento di Metodi e Modelli Matematici,
%Fac. Ingegneria, Universit\`a di Palermo, I-90128  Palermo, Italy}
\vspace{3mm}\\
 DEIM -Dipartimento di Energia, ingegneria dell' Informazione e modelli Matematici,
\\ Scuola Politecnica, Universit\`a di Palermo, I-90128  Palermo, Italy\\
e-mail: fabio.bagarello@unipa.it
\vspace{2mm}\\
{\large C. Trapani}
\vspace{3mm}\\
  Dipartimento di Matematica ed Informatica,
\\ Universit\`a di Palermo, I-90128  Palermo, Italy\\
e-mail: camillo.trapani@unipa.it
\vspace{2mm}\\
{\large S. Triolo}
\vspace{3mm}\\
  DEIM -Dipartimento di Energia, ingegneria dell' Informazione e modelli Matematici,
\\ Scuola Politecnica, Universit\`a di Palermo, I-90128  Palermo, Italy\\
e-mail: salvatore.triolo@unipa.it
\end{center}

\vspace*{2cm}

\begin{abstract}
\noindent Motivated by the growing interest on PT-quantum mechanics, in this paper we discuss  some facts on generalized Gibbs states and on their related KMS-like conditions. To achieve this, we first consider some useful connections between similar (Hamiltonian) operators and we propose some extended version of the Heisenberg algebraic dynamics, deducing some of their properties, useful for our purposes.
\end{abstract}

\vspace{2cm}

%{\bf PACS Numbers}:  .......

\vfill

\newpage

% Section 1

\section{Introduction and  notations }

In ordinary quantum mechanics a physical system $\Sc$ is described, first of all, by a self-adjoint Hamiltonian operator $H_0$. This is to ensure first that the energies of the system (i.e., the eigenvalues of $H_0$) are real, and secondly that the time evolution of $\Sc$ is unitary. If $\Sc$ lives in a finite dimensional Hilbert space $\Hil$, then $H_0$ is surely bounded. But, when $\dim(\Hil)=\infty$, quite often $H_0$ turns out to be unbounded. Hence, in physics, the role of unbounded operators is crucial in several cases.

In recent years, since the paper \cite{bb}, an increasing interest for systems described by non self-adjoint Hamiltonians, $H\neq H\*$ with only real eigenvalues, has spread among the physicists first, and the mathematicians later. We refer to \cite{PT1}-\cite{PT5}, and references therein, for two reviews and three recent volumes on this topic.
The reason is that, in some concrete applications in physics, it may happen that  three different operators, $H_0=H_0\*$, $H$ and $H\*$  (which, throughout the paper,   will be assumed to be closed and, at least, densely defined) have only point spectra, and that
all their eigenvalues coincide. In particular, in
\cite{baginbagbook} several triples of operators {of this kind} have been
discussed and the following eigenvalues equations have been {found in concrete quantum mechanical models} \be
H_0e_n=ne_n,\qquad H\varphi_n=n\varphi_n, \qquad
H\*\psi_n=n\psi_n, \label{21}\en where  { $\F_e=\{e_n\in\Hil,\,n\geq0\}$ } is an
orthonormal (o.n.) basis for the Hilbert space $\Hil$, while
$\F_\varphi=\{\varphi_n,\,n\geq0\}$ and
$\F_\psi=\{\psi_n,\,n\geq0\}$ are two biorthogonal sets,
$\left<\varphi_n,\psi_m \right>=\delta_{n,m}$, {but} not necessarily
bases for $\Hil$. However, quite often, $\F_\varphi$ and $\F_\psi$
are complete (or total: the only vector which is orthogonal to all the $\varphi_n$'s, or to all the $\psi_n$'s, is the zero vector) in $\Hil$ and, see \cite{baginbagbook}, they are also
$\D$-quasi bases, i.e., they produce a weak resolution of the
identity in a suitable set $\D$, dense in $\Hil$:
$$
\sum_{n}\left<f,\varphi_n\right>\left<\psi_n,g\right>=\sum_{n}\left<f,\psi_n\right>\left<\varphi_n,g\right>=\left<f,g\right>,
$$
for all $f,g\in\D$.

{As it is known}, the o.n. basis $\F_e$ can be used to define a
Gibbs state as follows: \be
\omega_0(X)=\frac{1}{Z_0}\sum_n\left<e_n,{\sf e}^{-\beta
H_0}Xe_n\right>, \label{22a}\en where $Z_0:=\sum_n\left<e_n,{\sf
e}^{-\beta H_0}e_n\right>=\sum_n {\sf e}^{-\beta n}=\frac{{\sf
e}^{\beta}}{{\sf e}^{\beta}-1}$ {and $\beta$ is the inverse
temperature, always positive}. Sometimes $\omega_0$ is
written as $\omega_0(X)= {tr}(\rho X)$, where
$\rho:=\frac{1}{Z_0}\,{\sf e}^{-\beta H_0}$ and $tr$ is the natural trace on $B(\Hil)$. Hence, in view of (\ref{21}), it is  interesting to see what can be done if, in (\ref{22a}), we replace $H_0$ with $H$ or with $H\*$, and the $e_n$'s with the $\psi_n$'s or with the $\varphi_n$'s.

For this reason, in this paper a particular attention is devoted to these Hamiltonian
operators and to their roles in Gibbs-like states:  some of their properties are derived and some
examples are discussed. After this preliminary analysis, we will explore how they can be used to define
some (generalized) Gibbs states, extending in different ways equation (\ref{22a}).

The paper is organized as follows. In Section 2, after some
preliminaries, we discuss the mathematical  {settings} and the functional structure
associated to the operators  $H,$  $H_0$ and $H\*$ that
obey the eigenvalues equations in (\ref{21}). Several pathologies concerning their
structure {will be } considered. In particular we will prove that,
under certain conditions, $H_0$ and $H$ are {\em
similar} operators, in the sense od Definition \ref{def simile} below.

 In Section 3 we extend the Gibbs state $\omega_0$ to the general
 situation where we  {work with} $\F_\varphi$ and $\F_\psi$, rather than with
 $\F_e$, and where we replace $H_0$ with $H$ or with $H^*$. Moreover we will investigate if some (generalized version of) the KMS-relation is also satisfied by these states.  This analysis will force us to introduce different concepts of algebraic dynamics, driven respectively by $H_0$, $H$ and $H\*$.  In Section  4  we consider other possible
generalizations of
  $\omega_0$, and we deduce again the related KMS-like conditions.  Our concluding remarks are given in Section 5.

\section{Preliminary results}
To keep the situation more general, all throughout this section we replace \eqref{21} with the following eigenvalue equations
\be
H_0 e_n=\lambda_n e_n,\qquad H\varphi_n=\lambda_n\varphi_n, \qquad
H\*\psi_n=\lambda_n\psi_n, \label{211}\en
where the $\lambda_n$'n are real numbers
and  we will suppose, in particular, that the sets
$\F_\varphi=\{\varphi_n\}$ and $\F_\psi=\{\psi_n\}$ in (\ref{21}) are Riesz bases. This means that we can find a bounded operator $T$,
with bounded inverse, such that \be \varphi_n=Te_n,\qquad
\psi_n=(T\*)^{-1}e_n, \label{2.1}, \quad \forall n\in {\mathbb N}.\en

We will study, in this
particular situation, the  relation  between $  H$ and $
H_0$ and how this relation is connected with the operators ${  H}$
and $  H_0$ being similar to each other. This means that $H$ and $H_0$ satisfy the following definition (see, e.g. , \cite{anttrap}):

\begin{defn} \label{def simile}{\rm Let $(V, D(V)),\,$ $(K,
D(K))$  be two linear operators in the Hilbert space $\Hil.$
We say that $V$ and $K$  are {\em similar}, and write $V\sim K$, if
there exists a bounded operator $T$ with bounded inverse $T^{-1}$
which intertwines $K$ and $V$ in the sense that $T:D(K)\rightarrow
D(V)$ and $VTf=TKf,$ for every $f\in D(K).$ }
\end{defn}

The bounded operator $T$ of definition \eqref{def simile} is called
a bounded intertwining operator for (or between) $V$ and $K$, \cite{anttrap}. Intertwining operators have been proved to be quite important in the construction of many exactly solvable quantum models, both when the Hamiltonian of the system considered is self-adjoint, \cite{intop}, and when is not, \cite{bagIO}.

{Before proceeding we fix some notation: given a linear operator $S$, with domain $D(S)$ we denote by $S\upharpoonright \G$ the restriction of $S$ to a subspace $ \G\subseteq D(S)$.
Moreover, we recall that, if $S$ is a closed operator, a subspace $\C\subseteq
D(S)$  is called a {\em core} for $S$ if the closure of
$\overline{S\upharpoonright \C}={S}.$} %\fabioalert{Camillo: lo definiamo?}
{For notations and basic definitions on operators and their spectra we refer to \cite[Ch. VIII]{reedsimon}.}

\begin{prop} \label{prop2} Let $({  H}_0,
D({  H}_0))$ and   $({  H},
D({  H}))$ be closed operators in $\Hil$, with $H_0$ selfadjoint. Let $\{e_n\}$ be an orthonormal basis consisting of eigevectors of $H_0$ and $\F_\varphi$ and
$\F_\psi$ be Riesz bases  as in \eqref{2.1} for which \eqref{211} holds. Assume that \be T{
H}_0e_n={  H}Te_n, \qquad {  H}\*
(T\*)^{-1}e_n=(T\*)^{-1}{  H}_0e_n.\en
{Then, the following statements are equivalent.
\begin{enumerate}
\item ${  H}\sim {  H}_0$, with intertwining operator $T$.
\item The linear span of $\{T e_n\}$, $\spn\{Te_n\}$,  is a core for $  H$.
\end{enumerate}
 }
\end{prop}

\begin{proof}
 First of all, we notice that  $\spn\{e_n\}$
 is a core for ${  H}_0$, as is easy to check.

 We prove that 1 $\Rightarrow$ 2. Let $g \in D(H)$ and put $f:=T^{-1}g$. Since $f \in D(H_0)$, there exists a sequence $\{f_n\}\subset \spn\{e_n\}$ such that
 $f_n \to f$ and $H_0 f_n \to H_0 f$. We put $g_n=Tf_n$, $n \in {\mb N}$. Clearly $\{g_n\}\subset \spn\{Te_n\}$ and $g_n \to g$. Moreover, $Hg_n=HT f_n =TH_0f_n \to TH_0f=HTf=Hg$,
proves the statement.

 Now, we show that 2 $\Rightarrow$ 1.  Since $\spn\{e_n\}$
 is a core for $H_0$, for every $f\in
D({  H}_0),$ there exists  a sequence $\{f_k\}\subset \spn
\{e_n\}$ such that
 $f=\lim_n{f_k},$  ${  H}_0f=\lim_n{{  H}_0f_n}$
and it is clear that  $T{  H}_0f=\lim_k{T{
H}_0f_k}=\lim_k{ {H}Tf_k}.$ Since, by hypothesis, $Tf\in
D(  H),$ then $TD({  H}_0)\subseteq D(  H)$ and $T{
H}f={  H}Tf$ for every $f\in D({  H}_0).$ Moreover if
$g\in D(  H),$ by hypothesis, there exists $\{
g_k\}\subseteq \spn \{\vp_n \} $ such that \be g=\lim_k
g_k \quad \mbox{and} \quad {  H}g=\lim_k {  H}g_k.
\label{217d}\en If we put
$f:=T^{-1}g=\lim_kT^{-1}g_k:=f_k\in \spn\{e_n\}$ we have
$$\lim_k {  H}_0f_k=\lim_k{  H}_0T^{-1}g_k=\lim_kT^{-1}{  H}g_k={  H}g.$$
Hence $f\in D(\overline{{  H}_0\upharpoonleft \spn\{\vp_n\}})$
and ${  H}_0f={  H}g.$ We deduce that $Tf_k=g_k,$
$f_k\in \spn \{e_n\} $ which implies by \ref{217d} that ${
H}g=\lim_k {H  }g_k=\lim_k {  H}Tf_k=\lim_kT{{
H}_0}f_k$ and $\lim_k {  H}_0f_k=T^{-1}{  H}g.$ Hence
$\lim_kT^{-1}g_k=\lim_kf_k=f$ then $f\in D({  H}_0)$
and $\lim_k{  H}_0f_k=H_0f.$ Therefore
$$ T{  H}_0f=\lim_kT{  H}_0f_k=\lim_k{  H}Tf_k={  H}g={  H}Tf $$
\end{proof}

%\myalert{[Ometterei il remark: \`e gi\`a nell' enunciato della prop. precedente. L' ho messo in commento.]}
%\noindent {\bf Remark--} We observe that the intertwining operator for ${  H}$ and ${
%H_0}$ is exactly the operator $T$ appearing in definition \eqref{2.1} of the Riesz
%bases.

\begin{prop}
If ${  H} \sim {{  H}_0},$ {with intertwining operator $T$}, and  $\F_\varphi$ and $\F_\psi$ are
Riesz bases  then, \be
 D({  H})=\{g\in\Hil: \quad \sum_k |\lambda_k|^2 |\ip{g}{\psi_k}
|^2<\infty \}. \label{217a} \en
\end{prop}

\begin{proof} Indeed
\begin{eqnarray*}
D({{  H}}) &=&\{Tf; \quad f\in D({{  H}_0})\}= \{ g\in\Hil: \quad T^{-1}g\in D({{  H}_0}) \}\\
 &=& \{ g\in\Hil: \quad \sum_k |\lambda_k|^2 |\ip{T^{-1}g}{e_k} |^2<\infty  \}\\
 &=&\{ g\in\Hil: \quad \sum_k |\lambda_k|^2 |\ip{g}{(T^{-1})\*e_k} |^2<\infty  \}\\
 &=&\{ g\in\Hil: \quad \sum_k |\lambda_k|^2 |\ip{g}{\psi_k} |^2<\infty  \}.
\end{eqnarray*}
\end{proof}

Now is clear that, {for every $f' \in D(H)$ and $g' \in \Hil$, \be
\ip{T{{  H}_0}T^{-1}f'}{g'}=\ip{{{  H}_0}T^{-1}f'}{T\*g'}=\sum
\lambda_k \ip{P_kT^{-1}f'}{T\*g'} \label{217b} \en where
$P_k\zeta:=\ip{\zeta}{e_k}e_k, \; \zeta \in \Hil. $

Then, if we put $R_k:=TP_kT^{-1}$, it is easy to check that
$\ip{{  H}f'}{g'}=\sum \lambda_k \ip{R_kf'}{g'}.$ Then
$\{R_k\}$ is a (non-self-adjoint) resolution of the identity, {in the sense that $R_kR_j=\delta_{k,j}R_j$ and $\sum_{k}R_kf =f$, for every $f \in \Hil$.}
%\fabioalert{Insisto: ma dove sono presi $f'$ ed $g'$? Magari \`e una risoluzione dell'identit\`a ma solo su un denso? Stesso problema con $u(H)$ sotto}

If $u$ is a {bounded} continuous function, one can define $\ip{u(
{H})f'}{g'}:=\sum_k u(\lambda_k)\ip{R_kf'}{g'} $, { $f', g' \in \Hil$}. Hence,
$$u({  H})=Tu({{  H}_0})T^{-1},$$ (for instance
$u({  H}):$=${  \e}^{i {{  H}}t}$=$T$ $\e^{i{ {  H}}_0t}
T^{-1}$).}

Now we set $$V(t):=T\e^{i {  H}_0t}T^{-1}, \quad t\in {\mathbb R}.$$ Then, \be V(0)=\1, \quad
\quad V(t+s)=V(t)V(s).\label{217c} \en Notice that $V(0)=V(t)V(-t),$ hence
$V(-t)=V(t)^{-1}.$  The boundedness of $T$ and $T^{-1}$ implies that  $ \| V(t)\|\leq  \|T  \|
\cdot\|T^{-1}  \|,$ for all
$t\in {\mathbb R}.$
{Hence $V(t)$ is a {\em  ({uniformly bounded}) one-parameter group of bounded operators}.

{\begin{prop} Let $\{ W(t), \, t\in {\mathbb R}\}$ be a  one-parameter
group of bounded operators %that satisfies the properties
%\eqref{217c} and
such that $t\rightarrow W(t)$ is strongly
continuous.
Set $$D({{  H}}): =\left\{f\in\Hil:
\lim_{t\rightarrow 0}\frac{W(t)-1}{t}f \quad \mbox{ exists in
} \Hil\right\},$$
$$ {{  H}}(f):=\lim_{t\rightarrow 0}\frac{W(t)-1}{it}f. $$
Then, the following statements are equivalent.
\begin{enumerate}
\item There exists a bounded operator $T,$ with bounded inverse,
such that $\{ TW(t)T^{-1}:\, \, t\in {\mathbb R}\}$ is a unitary group.
\item {There exists a self-adjoint operator $K$ such that
${  H}$ $\sim {  K}$.}
\end{enumerate}
\end{prop}
}
\begin{proof}
$1\Rightarrow 2$. Let $U(t):=T^{-1}W(t)T, \, t\in {\mathbb R}$ be
a one-parameter unitary group generated by {the self-adjoint operator $K$ defined as  follows}
$${  K}:=\lim_{t\rightarrow 0}\frac{U(t)-1}{it}f, \quad\quad
D({  K}):=\left\{f\in\Hil: \lim_{t\rightarrow 0}\frac{U(t)-1}{t}f \quad
\mbox{  exists in } \Hil\right\}.$$ Then $TD({  K})=D({  H})$ and ${  H}Tf=T{  K}f,$
for every $f\in D({  K}).$

Indeed,  for every $f\in\D({  K})$
$$Kf=\lim_{t\rightarrow 0}\frac{U(t)-1}{it}f=\lim_{t\rightarrow 0}\frac{T^{-1}W(t)T-T^{-1}T}{it}f=T^{-1}\lim_{t\rightarrow 0}\frac{W(t)Tf-Tf}{it}.$$
Then, $Tg\in\D({  H})$ and ${  H} T f=T{  K}f.$

Conversely, if
$h\in\D({  H})$

$${  H}h=\lim_{t\rightarrow 0}\frac{W(t)-1}{it}h=\lim_{t\rightarrow 0}\frac{T[W(t)-1]T^{-1}}{it}f=T\lim_{t\rightarrow 0}\frac{W(t)-1}{it}T^{-1}h.$$
Then, $T^{-1}h\in D({  K})$; i.e., $h\in TD({  K})$ and ${  H}h=T{  K}T^{-1}h.$
Moreover ${  H}Tf=T{  K} f$, for every $f\in D({  K}).$

{$2\Rightarrow 1$. Assume that $H\sim K$, with $K$ self-adjoint and intertwining operator $T$ (i.e., $TD(K)=D(H)$, $HTf=TKf$, for every $f \in D(K)$). Let $U(t)={\sf e}^{iKt}$, $t \in {\mb R}$ and define
$W(t)= TU(t)T^{-1}$ $t \in {\mb R}$. Then it is easily seen that $W(t)$ is a strongly continuous one-parameter group satisfying \eqref{217c}.
}
\end{proof}

% {perch\'e riscrivere queste cose?}\myalert{Let  $H,H_0, H\*$ be always as in \eqref{1.3} and
%$  H, {  H}_0,  {  H}\*$ be their respective closures with
%${  H}_0$ self-adjoint. We assume that
%
%$$\sigma({  H_0})=\sigma_p
%({  H_0} )=\{\lambda_n; \, n\geq 0\}$$ and that the corresponding
%eigenvectors $\{ e_n \}$ constitute an o.n. basis for $\Hil.$}

{The
assumption ${  H}Te_n=T{  H_0}e_n$ implies that ${  H}T
e_n= {\lambda_n}Te_n.$ {Hence, for the set of eigenvalues (i.e. the point spectra), we get  $\sigma_p({
H_0})\subseteq\sigma_p({  H})$ and since, by assumption, the whole spectrum $\sigma(H_0)$ coincides with $\sigma_p(H_0)$}, we conclude that $\sigma(H_0) \subseteq \sigma_p(H) \subseteq \sigma(H)$.
%\fabioalert{Camillo: li definiamo? O e' ovvio di che cosa si stia parlando}
If $H\sim H_0$, the converse inclusions hold. Indeed, in that case, as shown in \cite[Proposition 3.1]{anttrap}
similarity preserves not only the spectrum set $\sigma({
H})=\sigma({  H_0})$ but also the parts in which the spectrum is
traditionally decomposed: the point spectrum $ \sigma_p({
H})=\sigma_p({  H_0}) $, the continuous spectrum $\sigma_c({
H})=\sigma_c({  H_0})$ and the residual spectrum $\sigma_r({
H})=\sigma_r({  H_0}).$}
%($ \fabioalert{Does the converse inclusion
%hold? This is what we will explore next.

%The similarity of ${  H} $ and ${  H_0}$, preserves the
%spectra of the operators.} \fabioalert{\bf Le due frasi precedenti mi sembrano un p\'o contraddittorie, a meno che non stiamo parlando, nella seconda, di $\sigma$ e non di $\sigma_p$. Forse andrebbe esplicitato} But, in general, it does not preserve
%self-adjointness. Indeed,
{Therefore, if
$\F_\varphi$ and $\F_\psi$ are Riesz bases  satisfying the assumption of Proposition \ref{prop2} and if
%$({  H}_0, D({  H}_0))$ is a self-adjoint operator such that
%$$T{  H}_0e_n={  H}Te_n, \qquad {  H}\*
%(T\*)^{-1}e_n=(T\*)^{-1}{  H}_0e_n,$$
% and if
the linear span of $\{\vp_n\}$ is a core for $  H$,
then $\sigma({  H_0})=\sigma_p({  H_0})=\sigma_p({  H})=\sigma({  H}).$

%\medskip

%\myalert{NOTA: La proposizione 5 va omessa (tutta la parte in magenta): si \`e gi\`a detto poche righe sopra che, se $H\sim H_0$, $\rho({  H})=\rho({  H_0})$, dando pure la referenza. Basta osservare che per la Prop. 2, se $\spn\{Te_n\}$ \`e un core per $H$, allora  ${  H}$, allora $H\sim H_0$!!!! Direi di modificare come segue.}

%\medskip

%But there is more. By Proposition \ref{prop2}, if $\spn\{Te_n\}$ is a core for ${  H}$, then $H \sim H_0$ and, as mentioned before, $\rho({
%H})=\rho({  H_0})$. Hence, $\sigma({  H})=\sigma_p({  H})=\sigma_p({
%H_0})=\sigma({  H_0}).$}

%Now we go one step
%further and prove that, always under the assumption that $\spn\{Te_n\}$ is a core for ${  H}$,
%the whole spectrum of ${  H}$ coincides with $\sigma_p({
%H_0}),$ thus $\sigma({  H})=\sigma_p({  H})=\sigma_p({
%H_0})=\sigma({  H_0}).$}

Of course, what we have obtained here for $H$ can also be deduced for $H\*$, under similar assumptions. We postpone the analysis of the relation between $\sigma(H_0)$, $\sigma(H)$ and $\sigma(H\*)$ when $T$ or $T^{-1}$ is unbounded to a future paper.

\section{Generalizing Gibbs states}

Let us now go back to the Gibbs states introduced in (\ref{22a}).
  The properties of
$\omega_0$ are well known in the literature, both from a
mathematical and from a physical side. In particular, $\omega_0$
is linear, normalized, continuous and positive:
$\omega_0(X\* X)\geq 0$, {for every $X\in B(\Hil)$},and in particular $\omega_0(X\* X)= 0$  only when
$X=0$.

 Moreover, if we define the following standard Heisenberg
time evolution\footnote{The reason why we use the word standard
here is because in the following we will propose different
definitions for the time evolution of an operator, useful when
the dynamics is driven by some non self-adjoint Hamiltonian, as in our case.} on
$B(\Hil)$ %the set of all bounded operators on $\Hil$,
by \be
\alpha_0^t(X):=\e^{i{  H_0}t}X\e^{-i{  H_0}t},\quad X\in B(\Hil), \label{add1}\en
it follows that $\alpha_0^t(X)\in B(\Hil)$, {for every $t\in {\mathbb R}$}, and $\omega_0$  satisfies the following
equality: \be
\omega_0\left(A\alpha_0^{i\beta}(B)\right)=\omega_0(BA),\quad \forall A,B\in B(\Hil)
\label{23}.\en  The abstract version of
this equality is known in the literature as the {\em KMS
{condition}}, \cite{br2}, and it is used to analyze  {physical aspects of the system under investigation like, for instance,} its phase
transitions  {and its thermodynamical equilibria}. The parameter $\beta$, as already stated, is interpreted as the inverse
temperature of the system.

%\subsection{Two possible generalizations}
{What we are interested in here is the possibility} of defining {states of the type} $\omega_0$ {by using} the sets $\F_\varphi$ and $\F_\psi$, rather than $\F_e$, and by replacing $H_0$ with $H$ or with $H^*$, see (\ref{21}).

{We will not suppose in this section that $\F_\varphi$ and $\F_\psi$ are necessarily Riesz bases, except when explicitly stated, but only that they are biorthogonal sets. {We maintain the assumption that they are families of eigenvectors of $H$ and $H^*$, respectively, as in \eqref{21}. This more general situation is relevant in pseudo-hermitian quantum mechanics, \cite{baginbagbook}, and for this reason we believe it can be useful in concrete applications.} For the sake of simplicity, we assume that the eigenvalues of $H_0$, $H$ and $H\*$,
are equal to the sequence $\{n\}$ of natural numbers. The results of this paper remain mostly valid if the eigenvalues constitute a more general sequence $\{\lambda_n\}$, as in \eqref{211}, with just obvious modifications. }

We start with defining the following functionals:
\be \left\{\begin{array}{l}
\omega_{\varphi\varphi}(X)=\frac{1}{Z_{\varphi\varphi}}\sum_n\left<{  \e}^{-\beta {  H}}\varphi_n,X\varphi_n\right>=\frac{1}{Z_{\varphi\varphi}}\sum_n{  \e}^{-\beta n}\left<\varphi_n,X\varphi_n\right>,\\
{}\\
\omega_{\psi\psi}(X)=\frac{1}{Z_{\psi\psi}}\sum_n\left<{\sf e}^{-\beta {
H}\*}\psi_n,X\psi_n\right>=\frac{1}{Z_{\psi\psi}}\sum_n{\sf e}^{-\beta n}\left<\psi_n,X\psi_n\right>,\end{array}\right.
\label{24}\en where $X$,

 for the time being, is just an operator on $\Hil$ such that the right-hand sides above both converge, and $Z_{\varphi\varphi}=\sum_n{  \e}^{-\beta
n}\|\varphi_n\|^2$ and $Z_{\psi\psi}=\sum_n{  \e}^{-\beta
n}\|\psi_n\|^2$. Our main assumption on $\F_\varphi$ and $\F_\psi$, other than their completeness, which will always assumed all throughout the paper, and which is always satisfied in concrete examples in our knowledge,  is that also the series defining the normalizations above do
converge:

\begin{defn}{\rm The biorthogonal sets $\F_\varphi$ and $\F_\psi$ are called {\em well-behaved} if $Z_{\varphi\varphi}<\infty$ and $Z_{\psi\psi}<\infty$.}
\end{defn}

\vspace{2mm}

{\bf Remarks--} (1) It is easy to check that any two biorthogonal
Riesz bases $\F_\varphi$ and $\F_\psi$ are well-behaved. This implies, in
particular, that any o.n. basis $\F_e=\{e_n, \,n\geq0\}$ produces a
pair $(\F_\varphi=\F_e,\F_\psi=\F_e)$ of well-behaved biorthogonal sets (WBBS, from
now on). It is also not hard to find more examples of
biorthogonal sets $\F_\varphi$ and $\F_\psi$ which are not Riesz
bases, but still are well-behaved. For instance, if
$\{c_n,\,n\geq0\}$ is a sequence of real numbers such that $\sum_nc_n^{\pm 2}\e^{-\beta n}<\infty$, and if $\F_e$ is an o.n.
basis, then defining $\varphi_n=c_n e_n$ and
$\psi_n=\frac{1}{c_n}\,e_n$, the sets $\F_\varphi$ and $\F_\psi$
are biorthogonal and well behaved, even if they are not
Riesz-bases, which is what happens if $\{c_n\}$ or $\{c_n^{-1}\}$ diverges with $n$. As a concrete example, taking $c_n=1+n$ we have
$\F_\varphi=\{\varphi_n=(1+n)e_n,\,n\geq0\}$ and
$\F_\psi=\{\psi_n=\frac{1}{1+n}\,e_n,\,n\geq0\}$ are well-behaved.
Hence, the set of WBBS is rather rich.

\vspace{1mm}

{(2) In the situation described by \eqref{211}
 we will call well-behaved  the sets $\F_\varphi$ and
$\F_\psi$ if they satisfy $\sum_n{  \e}^{-\beta
\lambda_n}\|\varphi_n\|^2<\infty$ and $\sum_n{  \e}^{-\beta
\lambda_n}\|\psi_n\|^2<\infty$.}

\vspace{2mm}

From now on we will always consider, when not explicitly stated,
that the sets $\F_\varphi$ and $\F_\psi$ used to define
$\omega_{\varphi\varphi}$ and $\omega_{\psi\psi}$ are
well-behaved.  We call $\A_\varphi$ (resp. $\A_\psi$), the set of all the operators on $\Hil$, not necessarily bounded, such that $\sum_n {  \e}^{-\beta n}\mid\ip{\vp_n}{X\vp_n}\mid<\infty$ (resp. $\sum_n{  \e}^{-\beta
n}\mid\ip{\psi_n}{X\psi_n}\mid<\infty$).

It is possible to see that, since $\F_\varphi$ and $\F_\psi$ are
WBBS, each bounded operator  {$A\in B(\Hil)$} belongs to both $\A_\varphi$ and
$\A_\psi$. Indeed we have, for instance,
$$
|\omega_{\varphi\varphi}(A)| \leq \frac{1}{Z_{\varphi\varphi}}\sum_n{  \e}^{-\beta n}|\left<\varphi_n,A\varphi_n\right>|\leq \frac{1}{Z_{\varphi\varphi}}\sum_n\,{  \e}^{-\beta n}\|\varphi_n\|\|A\varphi_n\|\leq \|A\|,
$$ and in the same way $|\omega_{\psi\psi}(A)|\leq  \|A\|$, for all such $A$'s.
This means, incidentally, that $\omega_{\varphi\varphi}$ and
$\omega_{\psi\psi}$ are both continuous on $B(\Hil)$. Furthermore,
it is easy to understand that $\A_\varphi$ and $\A_\psi$ are
"larger than" $B(\Hil)$ itself. To understand that this is really the
case, let us consider the case in which the so-called {\em metric
operators} $S_\varphi$ and $S_\psi$ are unbounded,
\cite{baginbagbook}, as it happens  {quite often in concrete physical situations}. These operators are defined as follows:
$$
D(S_\varphi)=\left\{f\in\Hil:\, \sum_{ n} \left<\varphi_{n},f\right>\,\varphi_{ n} \mbox{ exists in }  \Hil \right\},\, \mbox{ and }\, S_\varphi f=\sum_{ n} \left<\varphi_{n},f\right>\,\varphi_{ n}
$$
for all $f\in D(S_\varphi)$, and, similarly,
$$
D(S_\psi)=\left\{h\in\Hil:\, \sum_{ n} \left<\psi_{n},h\right>\,\psi_{ n} \mbox{ exists in }  \Hil \right\},\, \mbox{ and }\, S_\psi h=\sum_{ n} \left<\psi_{n},h\right>\,\psi_{ n},
$$
 for all $h\in D(S_\psi)$. We first observe that, see \cite{baginbagbook}, $D(S_\varphi)$ and $D(S_\psi)$ are both dense in $\Hil$ if $\F_\varphi$ and $\F_\psi$ are complete. Moreover, if $\F_\varphi$ and $\F_\psi$ are biorthogonal Riesz bases, then both $S_\varphi$ and $S_\psi$ are even bounded, so that $D(S_\varphi)=D(S_\psi)=\Hil$.
 However, when $\F_\varphi$ and $\F_\psi$ are not Riesz bases (but they are well-behaved and complete),
 $S_\varphi$ and $S_\psi$ are unbounded, and still they are such that $S_\varphi\in \A_\psi$ and
 $S_\psi\in\A_\varphi$. The conclusion is that $\omega_{\varphi\varphi}$ and $\omega_{\psi\psi}$ are defined not only on $B(\Hil)$, but also on larger sets. Moreover, they are both linear, normalized (i.e.
$\omega_{\varphi\varphi}(\1)=\omega_{\psi\psi}(\1)=1$) and
positive on $B(\Hil)$ (i.e $\omega_{\varphi\varphi}(X\*
X)\geq0$ and $\omega_{\psi\psi}(X\* X)\geq0$, for all $X\in
B(\Hil)$). Also, if $\F_\varphi$ is complete,
$\omega_{\varphi\varphi}(X\* X)=0$ implies that $X=0$, and
if $\F_\psi$ is complete, $\omega_{\psi\psi}(X\* X)=0$
implies that $X=0$. Of course, $\omega_{\varphi\varphi}$
  and
$\omega_{\psi\psi}$ satisfy all the properties of states on
$C\*$-algebras, see \cite{br1}. For example, the Cauchy-Schwarz
inequalities for these states look as follows:
$$
\left|\omega_{\varphi\varphi}(A\* B)\right|\leq \omega_{\varphi\varphi}(A\* A)\omega_{\varphi\varphi}(B\* B),\qquad \left|\omega_{\psi\psi}(A\* B)\right|\leq \omega_{\psi\psi}(A\* A)\omega_{\psi\psi}(B\* B),
$$
for all $A,B\in B(\Hil).$

Let us now introduce the following quantities:
\be
I_\varphi(X):=\max\left\{\sum_n{  \e}^{-\beta n}\|\varphi_n\|\|X\varphi_n\|,\sum_n{  \e}^{-\beta n}\|\varphi_n\|\|X\*\varphi_n\|\right\}
\label{25}\en
and
\be
I_\psi(X):=\max\left\{\sum_n{  \e}^{-\beta n}\|\psi_n\|\|X\psi_n\|,\sum_n{  \e}^{-\beta n}\|\psi_n\|\|X\*\psi_n\|\right\},
\label{26}\en
and let $B_\varphi$ (resp. $B_\psi$) be the set of all the operators $X$ on $\Hil$, bounded or not, such that $I_\varphi(X)<\infty$ (resp. $I_\psi(X)<\infty$).

These sets both surely include $B(\Hil)$, as can be easily
checked. Also, they are larger than $B(\Hil)$, since, even if  $S_\varphi$ is not bounded, it still belongs to $B_\psi$.
In fact, since $I_\psi(S_\varphi)=\sum_n\e^{-\beta
n}\|\psi_n\|\|\varphi_n\|$, using the Cauchy-Schwarz inequality we find {
$$
I_\psi(S_\varphi)=\sum_n\e^{-\beta
n/2}\|\psi_n\|\e^{-\beta
n/2}\|\varphi_n\|\leq \left(\sum_n\e^{-\beta
n}\|\psi_n\|^2\right)^{1/2}\left(\sum_n\e^{-\beta
n}\|\varphi_n\|^2\right)^{1/2}<\infty,
$$
since $\F_\varphi$ and $\F_\psi$ are well-behaved.}
Hence $S_\varphi\in B_\psi$. Analogously,  it
could be checked that $S_\psi\in B_\varphi$, either if  $S_\psi$ is bounded, or not.

Moreover, $B_\varphi\subseteq \A_\varphi$ and $B_\psi\subseteq \A_\psi$. In fact, if $A\in B_\varphi$, then
$$
|\omega_{\varphi\varphi}(A)|=\frac{1}{Z_{\varphi\varphi}}\left|\sum_n {  \e}^{-\beta n}\left<\varphi_n,A\varphi_n\right>\right|\leq \frac{1}{Z_{\varphi\varphi}}\sum_n {  \e}^{-\beta n}|\left<\varphi_n,A\varphi_n\right>|\leq \frac{1}{Z_{\varphi\varphi}} I_\varphi(A),
$$
so that $A$ also belongs to $\A_\varphi$. Similarly one can check
the other inclusion. Incidentally we observe that, because of our previous result,
we are recovering here the fact that $S_\varphi\in \A_\psi$  and  $S_\psi\in
\A_\varphi$. The reason why $B_\varphi$ and $B_\psi$ are introduced is because $\A_\varphi$ and $\A_\psi$ are not, in general, ideals for $B(\Hil)$. However, the following result can be proved, which is a weaker version of $\A_\varphi$ and $\A_\psi$ being ideals:

\begin{prop} If  $\F_\varphi$ and $\F_\psi$  are WBBS then:

 {1. if $A\in B(\Hil)$ and $X\in B_\varphi$,
then $AX, XA\in \A_\varphi$;

2.  if $A\in B(\Hil)$ and  $X\in B_\psi$,
then $ AX, XA\in \A_\psi$.}

\end{prop}

\begin{proof}

Let us take $A\in B(\Hil)$ and $X\in B_\varphi$. Then
$I_\varphi(X)<\infty$. Hence, since
$$|\left<\varphi_n,AX\varphi_n\right>|\leq \|A\*\|\|\varphi_n\|\|X \varphi_n\|,
$$
we deduce that
$$
\left|\omega_{\varphi\varphi}(AX)\right|\leq
\frac{1}{Z_{\varphi\varphi}}\,\sum_n\e^{-\beta
n}|\left<\varphi_n,AX\varphi_n\right>|\leq
\frac{\|A\|}{Z_{\varphi\varphi}}I_\varphi(X)<\infty,
$$
which implies that $AX\in \A_\varphi$.The
other statements can be proved in a similar way.

\end{proof}

So far, we have only assumed that $\F_\varphi$ and $\F_\psi$ are well-behaved. Let us now discuss what happens if they are related to an o.n. basis $\F_e=\{e_n\}$ via some, in general unbounded, operator. More explicitly, we assume here that $\varphi_n=Te_n$ and $\psi_n=(T\*)^{-1}e_n$, as in (\ref{2.1}), but at least one between $T$ and $T^{-1}$ is unbounded.  In this case, see \cite{baginbagbook}, the sets $\F_\varphi$ and $\F_\psi$ might still be $\D$-quasi bases, for some dense subset $\D$ of $\Hil.$ This means that for all $f, g\in \D$, the following equalities  hold: \be
\left<f,g\right>=\sum_{n\geq0}\left<f,\varphi_n\right>\left<\psi_n,g\right>=\sum_{n\geq0}\left<f,\psi_n\right>\left<\varphi_n,g\right>.
\label{29b} \en
Then, it is convenient to assume that $\D$ is  stable under the action of both $T$ and $T^{-1}$, and that $e_n\in\D$ for all $n$. Hence $\varphi_n, \psi_n\in\D$ as well. This condition is satisfied in several concrete situations, \cite{baginbagbook}. In this case  we deduce
 \be
 \omega_{\varphi\varphi}(X)=\frac{Z_0}{Z_{\varphi\varphi}}\,\omega_0\left(T\* X T\right), \qquad  \omega_{\psi\psi}(Y)=\frac{Z_0}{Z_{\psi\psi}}\,\omega_0\left( T^{-1} Y( T^{-1})\*\right),
\label{29c}\en
for all $X\in\A_\varphi$ and $Y\in\A_\psi$. If both $T$ and $T^{-1}$ are bounded, i.e. when $\F_\varphi$ and $\F_\psi$ are Riesz bases, from (\ref{29c}) we can deduce the following inequalities:
$$
\omega_{\varphi\varphi}(X)\geq\frac{1}{\|T\|^2}\,\omega_0\left(T\* X T\right), \qquad \omega_{\psi\psi}(Y)\geq\frac{1}{\|T^{-1}\|}\,\omega_0\left( T^{-1} Y( T^{-1})\*\right),
$$
for all $X\in\A_\varphi$ and $Y\in\A_\psi$.

\subsection{ {The dynamics and the KMS-condition}}\label{sectIII1}

Equations in (\ref{29c}) show that $\omega_{\varphi\varphi}$ and
$\omega_{\psi\psi}$ can be related to $\omega_0$. Since we know
that this state satisfies the KMS-condition (\ref{23}), we now
investigate if some (generalized version of the) KMS-relation is also
satisfied by our states. In this section, we will always assume
that (\ref{2.1}) is satisfied and that $T, T^{-1}\in B(\Hil).$ Hence we are dealing with Riesz bases. The starting point of our
analysis are the following relations:
$$
T{{  H}_0}e_n={  H}Te_n, \qquad {  H}\*
(T\*)^{-1}e_n=(T\*)^{-1}{{  H}_0}e_n,
$$
for all $n$, which imply also that, for all complex $\gamma$,

\be T{\sf  e}^{\gamma {   {H_0}}}e_n={\sf  e}^{\gamma {
H}}Te_n, \qquad {\sf  e}^{\gamma {  H\*}}
(T\*)^{-1}e_n=(T\*)^{-1}{\sf  e}^{\gamma {  H_0}}e_n,
\label{210}\en for all $n$. Similar relations were already deduced in Section II. Of course, these equalities can be
extended to the linear span of the $e_n$'s, which is dense in
$\Hil$. How it has been extensively discussed in \cite{bag1,bag2},
in presence of non self-adjoint Hamiltonians it is not clear how
the {\em quantum dynamics} of a given system should be defined.
Recalling that the dynamics is one of the main ingredient of the
KMS-condition, it is clear that we have to face with this problem
also here. Natural possibilities which extend that in (\ref{add1})
are the following
\be
\alpha_\varphi^t(X)={\sf  e}^{it{  H}}X{\sf  e}^{-it{  H}},
\qquad \alpha_\psi^t(X)={\sf  e}^{it{  H\*}}X{\sf e}^{-it{
H\*}},
\label{add2}\en
for some $X\in\A$, see \cite{bag1,bag2}. These are two different, and both absolutely
reasonable, definitions of the {\em time evolution} of the operator
$X$. However, it is evident that these definitions present some
problems. First of all, since ${  H}$ and ${  H}\*$ are not
self-adjoint and since they are, quite often, unbounded, their exponentials
should be properly defined. Moreover,
in general, domain problems clearly occur: even if $f\in D({\sf e}^{-it{  H}})$, and $X\in B(\Hil)$ it is not guaranteed that
$X{  \e}^{-it{  H}}f\in D({  \e}^{it{  H}})$, in fact.

For this reason it is convenient to define $\alpha_\varphi^t(X)$
in the following alternative way: \be \alpha_\varphi^t(X):=T\alpha_0^t(T^{-1}XT)T^{-1},
\label{2111}\en for all $X\in B(\Hil)$. It is
clear that the right hand side of this equation is well defined, since only bounded operators are involved here.
It is interesting to notice that $\alpha_\varphi^t$ has all the
nice properties of a dynamics, \cite{br1,sew}. In particular,
$$
\alpha_\varphi^t(\1)=\1, \quad \alpha_\varphi^t(XY)=\alpha_\varphi^t(X)\alpha_\varphi^t(Y), \quad \alpha_\varphi^{t+s}(X)=\alpha_\varphi^t\left(\alpha_\varphi^s(X)\right),\quad \alpha_\varphi^0(X)=X,
$$
for all $X,Y\in B(\Hil)$ and for all $t,s\in
{\mathbb R}$. Moreover, $B(\Hil)$ is stable under the action of
$\alpha_\varphi^t$, which is also invertible:
$\left(\alpha_\varphi^t\right)^{-1}=\alpha_\varphi^{-t}$. It is
also easy to compute the generator of $\alpha_\varphi^t$, and one
can check that this generator is ${  H}=T{  H_0}T^{-1}$. In
fact, since $
\lim_{t,0}\frac{1}{t}\left(\alpha_0^t(X)-X\right)=i[{  H_0},X]$
for all bounded $X$, with similar computations to those of Section II, we can deduce that

$$
\lim_{t,0}\frac{\alpha_\varphi^t(X)-X}{t}=\lim_{t,0}\left[T\left(\frac{\alpha_0^t(T^{-1}XT)-T^{-1}XT}{t}\right)T^{-1}\right]=
$$
$$
=T\left[\lim_{t,0}\left(\frac{\alpha_0^t(T^{-1}XT)-T^{-1}XT}{t}\right)\right]T^{-1}=iT[{
H_0},T^{-1}XT]=i[T{  H_0}T^{-1},X].
$$

\vspace{2mm}

Going back to \eqref{2111}, from \eqref{210} it follows that, on a
dense domain, $Te^{\pm it{  H_0}}=e^{\pm it{  H}} T$, so that,
 {under the same stability conditions on $\D$ required just before Section \ref{sectIII1}, we have  \begin{eqnarray*} T\alpha_0^t(T^{-1}XT)T^{-1}f&=& T{\sf e}^{it{  H_0}}T^{-1}XT{  \e}^{-it{   H_0}}T^{-1}f\\&=&{\sf e}^{it{  H}}TT^{-1}XTT^{-1}{  \e}^{-it{  H}}f={  \e}^{it{
H}}X{  \e}^{-it{  H}}f\in\D,
\end{eqnarray*} for all $f\in\D$,} so that we go back to the natural definition of
the dynamics proposed formally in (\ref{add2}). Now, using the properties of the
trace, we deduce that \be
\omega_{\varphi\varphi}(BA)=\frac{1}{Z_{\varphi\varphi}}tr\left({\sf e}^{-\beta {  H}}TT\* BA\right), \label{212}\en while \be
\omega_{\varphi\varphi}(A\alpha_\varphi^{i\beta}(B))=\frac{1}{Z_{\varphi\varphi}}tr\left({\sf e}^{-\beta H}BTT\* A\right), \label{213}\en for all $A,B\in
B(\Hil)$. Therefore, if $B$ commutes with
$TT\*$, $[B,TT\*]=0$, then \be
\omega_{\varphi\varphi}(BA)=\omega_{\varphi\varphi}(A\alpha_\varphi^{i\beta}(B)).
\label{214}\en It is interesting to notice that the role of $A$ in
the relevant assumption for (\ref{214}) to hold is absolutely not
relevant. We also observe that, in case we have $T=\1$, everything collapses to
the standard situation described at the beginning of Section II. This is because, in this case, $\varphi_n=\psi_n=e_n$, and $H_0=H=H^*$.

It is also interesting to rewrite what we have deduced in a different form. Calling $B_T:=(TT\*)^{-1}B(TT\*)$, which is surely well-defined and bounded, the state $\omega_{\varphi\varphi}$ satisfies the following equation:
\be
\omega_{\varphi\varphi}(A\alpha_\varphi^{i\beta}(B))=\omega_{\varphi\varphi}(B_TA),
\label{215}\en
which of course returns (\ref{214}) if $[B,TT\*]=0$. The conclusion is therefore that $\omega_{\varphi\varphi}$ satisfies a KMS-like condition with respect to $\alpha_\varphi^t$, but with two different (but deeply related) operators, $B$ and $B_T$. Again, the role of $A$ seems to be not so important.

\vspace{2mm}

Similar computations and similar considerations can be repeated
for $\omega_{\psi\psi}$. In this case we put \be
\alpha_\psi^t(X)=(T\*)^{-1}\alpha_0^t\left(T^{*}X(T\*)^{-1}\right)T^*,
\label{216}\en for all $X\in B(\Hil)$. Using (\ref{210}) this is formally equal to $e^{itH\*}Xe^{-itH\*}$. In this
case we can prove that \be
\omega_{\psi\psi}(A\alpha_\psi^{i\beta}(B))=\omega_{\psi\psi}({_TB}A),
\label{217}\en for all $A,B\in B(\Hil)$. Here
we have defined ${_TB}=(TT\*)B(TT\*)^{-1}$. Of course,
if $[B,TT\*]=0$, we get ${_TB}=B$ and we can repeat the same
considerations as above. Moreover, similar conclusions about the
nature of $\alpha_\psi^t$ as an algebraic dynamics can be
repeated, repeating similar computations as those we have sketched for
$\alpha_\varphi^t$. For instance, $\alpha_\psi^t(\1)=\1$ and $\alpha_\psi^t(XY)=\alpha_\psi^t(X)\alpha_\psi^t(Y)$, for all $X,Y\in B(\Hil)$.

\vspace{2mm}

{\bf A simple example:--} Let $a$ be the bosonic operator satisfying, in the sense of unbounded operators, the commutation rule $[a,a^*]=\1$. Let $H_0$ be the self-adjoint Hamiltonian obtained by taking the closure of $a^*a$. If $e_0$ is the vacuum of $a$, $ae_0=0$, we can define new vectors $e_n=\frac{(a^*)^n}{\sqrt{n!}}\,e_0$, $n\geq1$. Hence, as it is discussed in any textbooks in quantum mechanics, the set $\F_e=\{e_n\}$ is an o.n. basis for the Hilbert space $\Hil=\Lc^2(\Bbb R)$. For any orthogonal projection operator $P$, $P=P^*=P^2$, the operator $T=\1+iP$ is bounded with bounded inverse $T^{-1}=\1-\frac{i+1}{2}\,P$. In particular, we will consider here $P$ as the projection operator on the normalized vector $u$: $Pf=\left<u,f\right>u$, for all $u\in\Hil$. Hence, calling $\varphi_n=Te_n=e_n+i\left<u,e_n\right>u$ and $\psi_n=(T^{-1})^*e_n=e_n+\frac{i-1}{2}\left<u,e_n\right>u$, the sets $\F_\vp=\{\vp_n\}$ and $\F_\psi=\{\psi_n\}$ are biorthogonal Riesz bases. The state $\omega_0$ in (\ref{22a}) can be {\em deformed} in the ways discussed previously. In particular, the state $\omega_{\varphi\varphi}$ in (\ref{24}) can be rewritten in terms of $\omega_0$ as follows:
$$
\omega_{\varphi\varphi}(X)=\frac{1}{1+\omega_0(P)}\omega_0\left(X+i[X,P]+PXP\right),
$$
for all bounded $X$. In fact: $Z_{\varphi\varphi}=\sum_n{\sf e}^{-\beta n}\|\varphi_n\|^2 = \sum_n{\sf e}^{-\beta n}(1+|\left<u,e_n\right>|^2) = Z_0(1+\omega_0(P)) $, while $\left<\varphi_n, X\varphi_n\right> = \left<Te_n, XTe_n\right> = \left<e_n, Xe_n\right> + \left<e_n, i(XP-PX)e_n\right> + \left<e_n, PXPe_n\right>$, which implies the previous equation.

It is also interesting to compare the different definitions of time evolution introduced so far. In particular we consider two normalized vectors $\Phi$ e $\Psi$ in $\Hil$ and the rank one operator $Y$ defined as $Yf=\left<\Psi,f\right>\Phi$. We are now going to show that, already for this simple operator, the dynamics $\alpha_0^t$ and $\alpha_\varphi^t$ are different:
$$
\alpha_0^t(Y)f=\left(|\Phi_0(t)\left>\right<\Psi_0(t)|\right)f,
$$
where we have defined $\Phi_0(t)={\sf e}^{iH_0t}\Phi$, $\Psi_0(t)={\sf e}^{iH_0t}\Psi$, and $\left(|h_1\left>\right<h_2|\right)f=\left<h_2,f\right>h_1$, for all $h_1, h_2$ and $f$ in $\Hil$. On the other hand, we get
$$
\alpha_\vp^t(Y)f=\left(|\Phi_\vp(t)\left>\right<\Psi_\psi(t)|\right)f,
$$
where $\Phi_\vp(t)=(T^{-1})^*{\sf e}^{iH_0t}T^*\Phi$, $\Psi_\psi(t)=T{\sf e}^{iH_0t}T^{-1}\Psi$. Of course, if $H_0$ commutes with $T$ then $\alpha_\vp^t(Y)f=\alpha_0^t(Y)f$, equality which can be extended to all the operators of $B(\Hil)$. However, already for this particular choice of $Y$, $\alpha_\vp^t$ and $\alpha_0^t$ are different.

Incidentally we can use (\ref{210}) to see that, if $\Phi$ and $\Psi$ belong respectively to the linear span of the $\psi_n$'s and of the $\vp_n$'s, then $T{\sf e}^{iH_0t}T^{-1}\Psi={\sf e}^{iHt}\Psi$ and $(T^{-1})^*{\sf e}^{iH_0t}T^*\Phi={\sf e}^{iH^*t}\Phi$, simplifying in this way the expressions of $\Phi_\vp(t)$ and $\Psi_\psi(t)$.

\section{More generalizations of $\omega_0$}

What we have done in the previous section suggests to consider  two other linear functionals, defined mixing the roles of $\F_\varphi$ and $\F_\psi$ as follows:

\be \left\{\begin{array}{l}
\omega_{\varphi\psi}(X)=\frac{1}{Z_{\varphi\psi}}\sum_n{  \e}^{-\beta n}\left<\varphi_n,X\psi_n\right>,\\
{}\\
\omega_{\psi\varphi}(X)=\frac{1}{Z_{\psi\varphi}}\sum_n{  \e}^{-\beta n}\left<\psi_n,X\varphi_n\right>.\end{array}\right.
\label{31}\en

Here
$Z_{\varphi\psi}=\sum_n{  \e}^{-\beta n}\left<\varphi_n,\psi_n\right>=Z_0$ and $Z_{\psi\varphi}=\sum_n{  \e}^{-\beta n}\left<\psi_n,\varphi_n\right>=Z_0$. Of course, both these series converge, without any need of further assumption: biorthogonality of $\F_\varphi$ and $\F_\psi$ is enough. Of course, this does not mean that also $\sum_n{  \e}^{-\beta n}\left<\varphi_n,X\psi_n\right>$ converges, in general. Hence, we have to consider first of all the problem of the existence of $\omega_{\varphi\psi}$ and $\omega_{\psi\varphi}$. However, it is easy to check first that
\be
\omega_{\varphi\psi}(X)=\overline{\omega_{\psi\varphi}(X\*)},
\label{32}\en
%\myalert{NOTA: questa uguaglianza \`e un indizio che non si tratta di funzionali positivi\ldots . Se lo fossero dovrebbe essere $\omega_{\varphi\psi}=\omega_{\psi\varphi}$ }
for all $X\in\A$ for which both sides of this equation exist. Then, from now on, we will concentrate on $\omega_{\varphi\psi}$, since the analysis of $\omega_{\psi\varphi}$ can be traced back to that of  $\omega_{\varphi\psi}$.

Let then $\A_{\varphi\psi}$ be the set of all the operators $X$ on $\Hil$, not necessarily bounded,
  such that  {$|\omega_{\varphi\psi}(X)|<\infty$}. If $\F_\varphi$ and $\F_\psi$
are WBBS,  then $B(\Hil)\subset \A_{\varphi\psi}$. In
fact, in this case, we have
\begin{align}
&|\omega_{\varphi\psi}(X)|\leq \frac{1}{Z_0}\sum {  \e}^{-\beta n}\|\varphi_n\|\|X\psi_n\|\leq \frac{\|X\|}{Z_0}\sum_n\left({  \e}^{-\beta n/2}\|\varphi_n\|\right)\left({  \e}^{-\beta n/2}\|\psi_n\|\right)
\nonumber \\
& \leq\frac{\|X\|}{Z_0}\sqrt{\sum_n{  \e}^{-\beta
n}\|\varphi_n\|^2}\,\sqrt{\sum_n{  \e}^{-\beta
n}\|\psi_n\|^2}<\infty, \label{33}\end{align} because of the definition of
WBBS. { It is also easy to check that
$\A_{\varphi\psi}$ is really "larger than" $B(\Hil)$}.
{In fact, even if $S_\varphi$ is unbounded,
we obtain
$\omega_{\varphi\psi}(S_{\varphi})=\frac{1}{Z_0}\sum_n{\sf e}^{-\beta n}\|\varphi_n\|^2$, which is surely finite. Hence  $S_\varphi\in\A_{\varphi\psi}$.}
%even if$S_\varphi\notin B(\Hil)$, as we had to check.
It is clear that $\omega_{\varphi\psi}$ is a normalized linear functional, {i.e., $\omega_{\varphi\psi}(\1)=1$.}
{ By \eqref{33} it follows that $\omega_{\varphi\psi}$ is bounded. Then $\omega_{\varphi\psi}$ will be positive if, and only if,
$$\|\omega_{\varphi\psi}\|= \sup_{\|X\|\leq 1}|\omega_{\varphi\psi}(X)|=\omega_{\varphi\psi}(\1)=1.$$ }

{Consider, for instance, the case where $\F_\varphi$ and $\F_\psi$ are Riesz bases. Then, simple calculations show that, for every $X \in B(\Hil)$,
\begin{align}\label{34} \omega_{\varphi\psi}(X)&= \frac{1}{Z_0} tr({\sf e}^{-\beta H^*}X) = \frac{1}{Z_0} tr({\sf e}^{-\beta H_0}T\*X(T\*)^{-1})=\omega_0(T\*X(T\*)^{-1}), \\
\nonumber \omega_{\psi\varphi}(X)&= \frac{1}{Z_0} tr({\sf e}^{-\beta H}X)=\frac{1}{Z_0} tr({\sf e}^{-\beta H}T^{-1}XT)=\omega_0(T^{-1}XT).  \end{align}
As is Section 3 we find that our new "states" are related to the original one, $\omega_0$, and that the relation is explicitly provided by the operator $T$ and its relatives ($T^*$, $T^{-1}$ and ${T^*}^{-1}$).

Using \eqref{34}, we get
$$ |\omega_{\varphi\psi}(X)|\leq \| T\*X(T\*)^{-1}\|, \qquad |\omega_{\psi\varphi}(X)|\leq \|T^{-1}XT\|, \quad X \in B(\Hil).$$
Therefore, if $\|T\|\|T^{-1}\|=1$, both $\omega_{\varphi\psi}$ and $\omega_{\psi\varphi}$ are positive, but we do not know what happens if $\|T\|\|T^{-1}\|>1$.

}

%What it is not clear, at a first sight, is that it is
%also positive. However, since (\ref{33}) implies that
%$|\omega_{\varphi\psi}(X)|\leq M\|X\|$, with $X\in B(\Hil)$ and with $0<M<\infty$,
% defined out of (\ref{33}), it also follows that, at
%least if $M\leq1$, then $\omega_{\varphi\psi}$ is positive on
%$B(\Hil)$. This is because, in this case, we can also check that
%$\|\omega_{\varphi\psi}\|=\omega_{\varphi\psi}(\1)=1$.

%\myalert{NOTA: questo discorso mi confonde. La costante $M$ trovata non pu\`o essere $<1$, perch\'e
%$$1=\omega_{\varphi\psi}(\1)\leq \|\omega_{\varphi\psi}\|\leq M.$$ L'unica cosa che si pu\`o dire \`e che se $|\omega_{\varphi\psi}(X)|\leq \|X\|$, allora $\|\omega_{\varphi\psi}\|=\omega_{\varphi\psi}(\1)=1$ e $\omega_{\varphi\psi}$ \`e positivo. Forse la questione della positivit\`a meriterebbe un approfondimento.}\fabioalert{ Camillo: fai di queste righe quello che vuoi! Se le vuoi cancellare, fai pure. Cos\'i ce ne usciamo.}

%\myalert{Se, come temo, i funzionali definiti sopra non fossero positivi, sarebbe opportuno dirlo e dare una qualche motivazione per cui ci pu\`o essere interesse fisico a considerarli, visto che come Fabio non si stanca di ripetere, ci stiamo rivolgendo ad un publico physically minded. Io non riesco a trovarne alcuna di motivazione, ma pu\`o essere un mio limite\ldots} \fabioalert{ho aggiunto una pseudo-giustificazione nel Remark dopo la (\ref{36})}

 An extended version of the
KMS-condition can be deduced also in this case, but making
reference to the {\em time evolution} $ \alpha_{\psi}^t(X)={  \e}^{it{
H}\*}X{  \e}^{-it{  H}\*}, $ for all
$X\in B(\Hil)$, already introduced in (\ref{add2}). We get
 \be
\omega_{\varphi\psi}(BA)=\omega_{\varphi\psi}\left(A\alpha_{\psi}^{i\beta}(B)\right),
\label{36}\en and both sides are equal to
$\frac{1}{Z_0}tr\left({  \e}^{-\beta H\*}BA\right)$.
Similar conclusions can be deduced for $\omega_{\psi\varphi}$. In
this case, however, $\alpha_{\psi}^t$ should be replaced by $\alpha_{\varphi}^t$.

\vspace{2mm}

{\bf Remark:--} It is worth noting that the condition in (\ref{36}) looks much closer to the original KMS-condition than the ones obtained in the previous section. In particular, the operator $T$ plays no role here, while it was quite relevant in (\ref{215}) and in (\ref{217}). This depends, of course, on the particular definition of $\omega_{\varphi\psi}$, which involves both $\varphi_n$ and $\psi_n$, while this was not so neither for $\omega_{\varphi\varphi}$ nor for $\omega_{\psi\psi}$.

\vspace{3mm}

Those considered so far are not the only possible generalizations of the state $\omega_0$. Other, quite general, possibilities arise if we replace the  operators $\e^{-\beta H}$ and $\e^{-\beta H\*}$ in (\ref{24}) with some generic operator $A\*A$, for some fixed nonzero $A\in B(\Hil)$, not necessarily related to $\F_\varphi$ or $\F_\psi$. In particular, we define the following functionals on $B(\Hil)$:
\be\label{add3} \left\{\begin{array}{l}
\omega'_{\varphi\varphi}(X)=\frac{1}{Z'_{\varphi\varphi}}\sum_n \left<\varphi_n,A\*AX\varphi_n\right>,\\
{}\\
\omega'_{\psi\psi}(X)=\frac{1}{Z'_{\psi\psi}}\sum_n
 \left<\psi_n,A\*AX\psi_n\right>,\\
{}\\
\omega'_{\varphi\psi}(X)=\frac{1}{Z'_{\varphi\psi}}\sum_n
\left<\varphi_n,A\*AX\psi_n\right>,\\
{}\\
\omega'_{\psi\varphi}(X)=\frac{1}{Z'_{\psi\varphi}}\sum_n
\left<\psi_n,A\*AX\varphi_n\right>,
\end{array}\right.
\en where $Z'_{\varphi\varphi}=\sum_n \|A\varphi_n\|^2,$
$Z'_{\psi\psi}=\sum_n\|A\psi_n\|^2$,
$Z'_{\varphi\psi}=\sum_n\ip{A\varphi_n}{A\psi_n},$ and
$Z'_{\psi\varphi}=\sum_n\ip{A\psi_n}{A\varphi_n}$. Of course, all the normalization factors need to be non zero. This is granted, at least when $\F_\varphi$ and $\F_\psi$ are at least complete in $\Hil$. In fact, let, for instance, assume that $Z'_{\varphi\varphi}=0$. Then $A\varphi_n=0$ for all $n$. Hence $\left<A^*f,\varphi_n\right>=0$ for all $n$ and for all $f\in\Hil$. Hence, if $\F_\varphi$ is complete, $A^*f=0$ for all such $f$'s. Hence $A^*=0$ and $A=0$ as a consequence. Similarly, if $A\neq0$ then $Z'_{\psi\psi}>0$. As for $Z'_{\varphi\psi}$ and $Z'_{\psi\varphi}$, they are also strictly positive if $A\neq0$ is an Hilbert-Schmidt operator, at least when $\F_\varphi$ and $\F_\psi$ are Riesz bases. In fact, under these assumptions we can check that $Z'_{\varphi\psi}=Z'_{\psi\varphi}=tr(A^*A)$, where $tr$ is the natural trace on $B(\Hil)$, which is always bounded and strictly positive.

We need now that also the series in (\ref{add3}) all converge.  {This can be achieved by imposing some (further) condition on the
operator $A$.} Let $
\{e_n \in\Hil, n \geq 0 \}$ be an o.n. basis of $\Hil$ and $T\in B(\Hil)$ a
self-adjoint, invertible operator  such that
$\vp_n=Te_n$. Then we have
$\omega'_{\varphi\varphi}(X)=\frac{1}{Z'_{\varphi\varphi}}tr(X|T\*|^2|A|^2)$.

Indeed we have
\begin{eqnarray*}
\omega'_{\varphi\varphi}(X) &=&\frac{1}{Z'_{\varphi\varphi}}\sum_{n} \left<\vp_n,A\*A X \vp_n \right>=\frac{1}{Z'_{\varphi\varphi}}\sum_{n} \ip{Te_n}{ A\*A XTe_n} \\
 &=&\frac{1}{Z'_{\varphi\varphi}} \sum_{n} \ip{e_n}{ T\* A\*A X Te_n}=\frac{1}{Z'_{\varphi\varphi}}tr( T\* A\*A X T)=\frac{1}{Z'_{\varphi\varphi}}tr(X|T\*|^2|A|^2).
\end{eqnarray*}

A similar relation between $\omega'_{\psi\psi}$ and the trace can be deduced if $T^{-1}$ is bounded. We recall that when both $T$ and $T^{-1}$ are bounded, then $\F_\vp$ and $\F_\psi$ are Riesz bases. When this happens
{it is interesting to observe that the state $\omega'_{\vp\psi}$ turns out to be independent of $T$. This implies, among other things, that $\omega'_{\vp\psi}(X)=\omega'_{\vp'\psi'}(X)$ for all bounded $X$ and for any two pairs of biorthogonal Riesz bases $(\F_\vp, \F_\psi)$ and $(\F_\vp', \F_\psi')$.   In this case, in fact

\begin{eqnarray*}
\omega'_{\vp\psi}(X) &=&\frac{1}{Z'_{\varphi\psi}}\sum_{n}\ip{\vp_n}{A\*A X\psi_n}=\frac{1}{Z'_{\varphi\psi}}\sum_{n} \ip{Te_n}{A\*AX (T^{-1})\*e_n} \\
 &=& \frac{1}{Z'_{\varphi\psi}}\sum_{n} \ip{e_n}{ T\*A\*A X (T^{-1})\*e_n}=\frac{1}{Z'_{\varphi\psi}}tr(T\*A\*A X (T^{-1})\*)\\
 &=&\frac{1}{Z'_{\varphi\psi}}tr(A\*AX(T^{-1})\*T\*)= \frac{1}{Z'_{\varphi\psi}}tr(A\*AX),
\end{eqnarray*}
and $T$ plays no role at all, as stated. This result extends to the present context a well known result on traces, which are independent of the particular o.n. basis used to compute them.
}

\section{Concluding remarks}

In this paper we have discussed some peculiarities concerning operators which are similar and which can be self-adjoint or not, but all having real eigenvalues. We have proposed several generalizations of the notions of {\em the algebraic Heisenberg dynamics} and of {\em Gibbs states}, and we have discussed which kind of KMS-like relations appear out of them. This particular aspect of our research,  in view of its concrete applications to quantum mechanics, is particularly promising and a deeper understanding of the various KMS-conditions deduced in the paper is, in our opinion, needed. This can be useful also in view of a correct definition of the dynamics of systems driven by non self-adjoint Hamiltonians, following the analysis already undertaken in \cite{bag1} and in \cite{bag2}.

\section*{Acknowledgments}

 {The authors acknowledge partial support from GNFM, GNAMPA and from the Universit\`a di Palermo.}

\end{document}